\documentclass[letterpaper]{article} 
\usepackage{aaai2026}  
\usepackage{times}  
\usepackage{helvet}  
\usepackage{courier}  
\usepackage[hyphens]{url}  
\usepackage{graphicx} 
\usepackage{natbib}  
\usepackage{caption} 
\usepackage{algorithm}
\usepackage{algpseudocode}

\usepackage{newfloat}
\usepackage{listings}
\DeclareCaptionStyle{ruled}{labelfont=normalfont,labelsep=colon,strut=off} 
\floatstyle{ruled}
\newfloat{listing}{tb}{lst}{}
\floatname{listing}{Listing}
\usepackage{amsmath,amssymb,amsthm,amsfonts}
\usepackage{float}
\usepackage{subcaption}
\usepackage{xcolor}
\usepackage{multirow,array}
\usepackage{cleveref}
\usepackage{bm}
\usepackage{multicol}

\newtheorem{defn}{Definition} 
\newtheorem{lem}{Lemma} 
\newtheorem{asm}{Assumption} 
 
\newtheorem{thm}{Theorem} 
\newtheorem{remark}{Remark} 
\newtheorem{corollary}{Corollary}

\newcommand{\Prob}{\mathbb{P}}

\newcommand{\N}{\mathcal{N}}

\newcommand{\VaR}{\mathrm{VaR}}

\title{Escaping the Subprime Trap in Algorithmic Lending}
\author{
    Adam Bouyamourn and Alexander Tolbert
}
\affiliations{
}

\usepackage{bibentry}

\begin{document}

\maketitle

\begin{abstract}
Disparities in lending to minority applicants persist even as algorithmic lending finds widespread adoption. We study the role of \emph{risk-management} constraints, specifically Value-at-Risk ($\VaR$) and Expected Shortfall (ES), in inducing inequality in loan approval decisions, even among applicants who are equally creditworthy. Empirical research finds that disparities in the interest rates charged to minority groups can remain large even when loan applicants from different groups are equally creditworthy. We contribute an original analysis of 431,551 loan applications recorded under the Home Mortgage Disclosure Act, illustrating that disparities in data quality are associated with higher rates of loan denial and higher interest rate spreads for Black borrowers.
We develop a formal model in which a mainstream bank (low-interest) is more sensitive to variance risk than a subprime bank (high-interest). If the mainstream bank has an inflated prior belief about the variance of the minority group, it may deny that group credit indefinitely, thus never learning the true risk of lending to that group, while the subprime lender serves this population at higher rates. We call this ``The Subprime Trap'': an equilibrium in which minority lenders can borrow only from high-cost lenders, even when they are as creditworthy as majority applicants. Finally, we show that a finite subsidy can help minority groups \emph{escape the trap}: subsidies cover enough of the mainstream bank's downside risk so that it can afford to lend to,  and thereby learn the true risk of lending to, the minority group. Once the mainstream bank has observed sufficiently many loans, its beliefs converge to the true underlying risk, it approves the applications of minority groups, and competition drives down the interest rates of subprime loans.
\end{abstract}

\section{Introduction}
\label{sec:intro}
Algorithmic lending has grown rapidly as scalable ML methods achieve wide adoption among both existing lenders and new market entrants, bringing with them the possibility of significantly improving the fairness of financial decisions based on observable data about loan applicants \citep{Khandani2010,Bono2021,Bergetal2022,Remolina2022}. However, inequalities persist in both loan approval rates and interest rates charged to minority applicants versus white applicants, and a switch to algorithmic lending procedures does not necessarily improve outcomes on either metric \citep{Romei2013, liobait2017, Quillian2020, Giacoletti2021-ts,aliprantis2022dynamics,Bartlett2022, fuster2021}. These inequities exist against a background of historic discrimination in US retail banking to individuals \citep{Quillian2020}, businesses \citep{Branchflower2003}, and a long-standing ``racial wealth gap": persistent differences in median household wealth by ethnic group \citep{Charles2002,derenoncourt2023, althoff2024jim}. 
 
Notably, disparities in interest rates and loan approvals persist even when minority applicants have comparable credit scores to majority applicants \citep{Bayer2016,Popick2022, Crosignani2023}. This combination of facts creates a puzzle: are lenders selecting on observable risk factors? And if not, why, besides explicit discrimination, might they be failing to lend to minority applicants?

Several empirical findings help us to understand this puzzle. First, ethnic groups sort across lenders, with minorities more likely to accept loans from high-interest rate banks versus conventional lenders, even conditional on credit score. \citep{Bayer2016} write:

\begin{quote}
African-American and Hispanic borrowers tend to be more concentrated at high-risk lenders. Strikingly, this pattern holds for all borrowers even those with relatively unblemished credit records and low-risk loans... High-risk lenders are not only more likely to provide high cost loans overall, but are
especially likely to do so for African-American and Hispanic borrowers. These lenders are largely responsible for the differential treatment of equally qualified borrowers.
\end{quote}

Modern algorithmic lending systems operationalize risk assessment through machine learning models trained on historical data. Minority applicants disproportionately have ``thin credit files'' with limited payment history data, and this data sparsity directly manifests as higher prediction variance in ML models \citep{blattner2021how}. When banks face risk-management constraints, the higher predicted variance translates into higher rates of loan denial. Because denied applicants never generate repayment data, banks cannotthen learn that their variance estimates were inflated, creating a self-reinforcing equilibrium of exclusion.

It is this set of stylized facts that motivates our model. How do minorities with good credit scores end up with no choice but to accept high-cost loans? 
To answer this question, we consider the role of risk management constraints when banks have imperfect information about their applicant pool. Lending is a problem of imperfect information: banks do not know the exact probability that a given borrower will repay a loan, and so must use observable characteristics as a proxy for repayment probability \citep{Adams2009, Crawford2018}. Banks rely on credit scoring, which is subject to known biases, and has lower quality information about the creditworthiness of minority applicants \citep{blattner2021how}. Reliance on credit scores can therefore lead banks to hold inaccurate prior beliefs about minority applicants, which can sustain equilibria in which minority groups systematically lose out \citep{Loury-Coate, Kim_Loury_2018}. 

We study the effect of Value-at-Risk ($\VaR$) constraints on loan approval decisions \citep{Artzner1999, Jorion2000-ng}. Our analysis focuses on a formal model in which a mainstream, or ``low interest rate", bank exhibits heightened sensitivity to variance risk relative to a subprime, or ``high interest rate", bank. We show that when the mainstream bank has inflated prior beliefs about the variance repayments from the minority group, $\VaR$ constraints bind, which means that it will then systematically refrain from extending credit to that group. This reluctance then prevents the bank from updating its beliefs regarding the true risk, thereby locking minority borrowers into higher-cost subprime lending arrangements. We describe this mechanism as ``the subprime trap".

\subsection{Our Contributions}
First, we contribute an analysis of risk management constraints to the study of algorithmic fairness. We help address a puzzle of general interest: why might banks fail to extend loans even when they would otherwise be profitable? Second, we bring details of real-world banking practices into the study of algorithmic fairness.
We introduce a two-bank equilibrium framework in which the mainstream bank enforces a more stringent $\VaR$ requirement than its subprime counterpart. Under these conditions, an erroneous initial estimate of the variance of the minority group leads the mainstream bank to permanently refrain from lending to that group. Third, we formalize the resulting equilibrium (\Cref{thm:SubprimeTrap}) and demonstrate that the minority group is consequently confined to subprime loans, despite having the same average creditworthiness as the majority group. Fourth, we establish that the provision of a modest, finite subsidy, or partial guarantee, can resolve this equilibrium inefficiency. By adequately mitigating downside risk, the subsidy enables the main bank to lend and learn about the creditworthiness of the minority group. With an updated assessment of risk, the bank eventually meets its $\VaR$ requirements without continued external intervention, thereby extending mainstream credit at favorable rates.

This investigation contributes to the literature by highlighting how inaccurate risk metrics, in conjunction with risk-based capital constraints, can systematically exclude certain borrower groups from low-interest credit. Our results further imply that targeted subsidies can effectively rectify informational failures, yielding improved outcomes for both lenders and minority borrowers.

\subsection{Related Literature}

Our work contributes to several strands of literature, including algorithmic fairness in lending, the economics of discrimination, and the role of informational frictions in sustaining suboptimal market equilibria. In the domain of algorithmic fairness, recent research has examined how increasingly sophisticated data-driven methods in credit scoring may inadvertently perpetuate or exacerbate disparities \citep{Barocas2016, barocas2023fairness, Kumar2022}. While algorithmic approaches promise enhanced accuracy, empirical evidence indicates that minority applicants often face higher interest rates or are denied credit outright \citep{Bartlett2022, Crosignani2023}. 

Discrimination is a classical topic in economics, studied in a large theoretical and empirical literature \citep{Becker1971-ee,Phelps_1972,Lundberg_1983,Arrow1998, Heckman1998, Cowgill2019}. Building on this tradition, subsequent studies have explored how negative stereotypes or misperceptions of risk may translate into adverse outcomes for minority borrowers \citep{Fang2011, Loury-Coate}. In a manner analogous to these analyses, our work demonstrates that inflated beliefs about repayment variance can result in a self-reinforcing equilibrium whereby minority groups are persistently relegated to subprime lending markets.

A related literature examines the role of informational frictions and the corrective impact of policy interventions \citep{Akerlof1970,Akerlof2015-du, Thaler2016}. \citep{Cai_2020} studied contexts in which banks can acquire improved information about borrowers' creditworthiness through selective experimentation. \citep{donahue_barocas2021} study the trade-offs between solidaristic insurance policies, in which minorities are subsidized in accordance with risks that are causally related to their minority status, and actuarially-fair insurance policies, in which everyone pays premia equal to their marginal risk.  

Our work also relates to work on delayed feedback with respect to implementation decisions in machine learning settings \citep{liu2018delayed, Pagan_2023, Chaney_2018}. We describe a setting in which mainstream banks are stuck in a negative feedback loop: their failure to learn about applicant creditworthiness is self-sustaining. Our work is also related to literature on suboptimal outcomes with bandits. \citep{honda2013optimalitythompsonsamplinggaussian} prove that Thompson Sampling's optimality depends critically on prior specification, while \citep{ghosh2017misspecifiedlinearbandits} show that any algorithm achieving optimal performance under perfect model specification must suffer linear regret under misspecification. This mirrors our setting where banks with incorrect variance priors never explore lending to certain groups. The dueling bandits framework of \citep{Yue2012} captures competition between lenders, establishing information-theoretically optimal algorithms for pairwise comparisons. \citep{Frazier_2014} study how to incentivize exploration when arms are pulled by self-interested agents -- our approach, in which subsidies encourage banks to `learn through lending', is closely related to this setting. 

Finally, our work contributes to a broader effort to model rational but suboptimal decision-making by or with respect to disadvantaged groups.\citep{diana2024adaptivealgorithmicinterventionsescaping} model pessimism traps and herding behavior as individuals grapple with cycles of noisy and censored information when making decisions about potentially higher reward but riskier ends. Our argument is similar in broad outline to that of \citep{FOSTER1992, Hu_2018}, who argue that temporary interventions in the labor market to address discrimination may lead to long-run improvements in the fairness of observed decisions. 

\section{Empirical Motivation}

We analyzed 2024 Home Mortgage Disclosure Act (HMDA) data, examining 431,551 conventional home purchase mortgage applications by single applicants listed as Black or White on loan applications. The HMDA requires that financial institutions that report at least 60,000 applications and loans in the previous calendar year must report data on lending applicants and loan outcomes. We examine racial disparities in two lending outcomes: loan denial decisions ($N = 431,551$) and interest rate spreads conditional on approval ($N = 247,329$).

\begin{table}[h]
\centering
\caption{Lending Outcomes by Race and Data Quality}
\begin{tabular}{lccc}
\hline\hline
\multicolumn{4}{l}{\textbf{Panel A: Denial Rates}} \\
\hline
 & Complete Data & Missing Income & Difference \\
\hline
Black & 26.1\% & 57.6\% & +31.5pp \\
White  & 9.8\% & 23.1\% & +13.3pp \\
\hline
Differential  & & & +18.2pp \\
\hline
& & & \\
\multicolumn{4}{l}{\textbf{Panel B: Interest Rate Spreads\footnote{Rate spread measures how much more expensive a borrower's loan is compared to the best rates available to prime borrowers at the time of origination. A spread of 1\% means the borrower pays 1 percentage point more in annual interest than the benchmark rate.} On Approved Loans}} \\
\hline
 & Complete Data & Missing Income & Difference \\
\hline
Black  & 0.85\% & 1.12\% & +0.27pp \\
White & 0.62\% & 0.74\% & +0.12pp \\
\hline
Differential  & & & +0.15pp \\
\hline\hline
\end{tabular}
\label{tab:lending_outcomes}
\end{table}

This analysis highlights that Black applicants with missing income data  face higher loan denial rates both marginally, and conditional on having missing data; and, if approved, pay higher interest rate spreads both marginally, and conditional on having missing income data. 

While this analysis is descriptive, it suggests a puzzle: why do differences in data quality lead to higher rejection rates and higher interest rates for Black applicants specifically? We propose an explanation for this observed phenomenon in our formal model below. 

\section{Preliminaries and Model}
\label{sec:prelim}

We analyze a multi-period environment indexed by 
$t = 1, 2, \ldots, T$. At the 
onset of each period $t$, the high-interest (subprime) bank selects an interest-rate 
premium $\nu_t \in [0, \nu^{\text{max}}] $. Each loan applicant belongs to one of two groups and applies simultaneously to both the mainstream 
(low-interest) bank and the subprime bank. The mainstream bank charges a baseline 
rate normalized to $1$, while the subprime bank charges $1 + \nu_t$.

Each bank evaluates whether to approve or reject each applicant in that period, 
seeking to maximize its expected profit subject to a risk-management constraint. If an applicant receives approvals from both banks, the applicant typically accepts the offer with the lower interest 
rate; if only one bank approves the loan, it is accepted. After the loans are finalized, each bank privately observes its payoff. Consequently, a bank that  rejects an application or whose offer is declined does not observe the repayment 
outcome for that applicant.

At the end of period $t$, both banks update their beliefs regarding the repayment 
distribution for each group, employing a Bayesian learning procedure described in 
Section~\ref{sec:updating}. These updated beliefs affect the subprime bank's 
subsequent choice of premium $\nu_{t+1}$ and both banks' lending decisions in 
period $t+1$. 

\subsection{Loan Applicants}

\subsubsection{Payoff distributions}

\label{sec:applicants}
There are two groups, denoted by $i \in \{W,B\}$. $W$ represents the racial majority, while $B$ represents the racial minority. Each group has a payoff distribution $\pi_{it} \sim \N(\mu_i,\sigma^2_i)$. The bank's return from approving a loan depends on this realized payoff. 
\begin{asm}[Equal expected creditworthiness]
$$\forall t: \mathbb{E}\left[\pi_{Wt}\right] = \mu_W = \mu_B = \mathbb{E}\left[\pi_{Bt}\right] > 0$$
\end{asm}
\noindent First, we assume that both groups have the same expected creditworthiness. We do this because we want to study the specific case where minority and majority racial groups are equally creditworthy.  We allow their variances, $\sigma_W^2$ and $\sigma_B^2$, to differ.

Second, we assume that this is \textit{common knowledge}: both banks know that each group is equally creditworthy. 

We suppose that each group strictly prefers to accept a loan rather than not, and strictly prefers to accept the cheapest loan. 
We assume that applicants randomize if both loans have the same price. We denote each group's decision to accept the loan with $S_i \in \{0,1\}$.
\subsubsection{Observable Applicant Characteristics}

We suppose that individuals are associated with repayment-probability-relevant applicant characteristics: repayment history, assets, income, net debt, and so forth. We suppose that this information is differentially missing by group: minority lenders have fewer complete characteristics. We formalize this simply by supposing that for each group, the number of complete applicant files, or the \textit{credit history for group }$i$, $C_i$ is binomial, so that we observe $np_i$ complete applicant files and $n(1-p_i)$ incomplete applicant files. 

\begin{asm}[Differential Information about Applicants]
    $$p_B < p_W \text{ so that }\mathbb{E}[C_B] < \mathbb{E}[C_W]$$
\end{asm}

\noindent For simplicity, we set $p_W = 1$, and $p_B < 1$. This is WLOG, because what matters for the derivation of our results is the difference in information between the groups. 

\subsection{Lenders}

The model features two banks: Bank $L$, which offers low-interest (mainstream) loans, and Bank $H$, which provides high-interest (subprime) loans. We refer to  banks with the index $j$. In each period, both banks decide whether to extend a loan to each group. $A_{ijt} \in \{0,1\}$ denotes whether bank $j$ approves a loan to group $i$ at time $t$. 
\subsubsection{Bank payoff functions}

If the bank issues a loan to an applicant from group $i$ in period $t$, it earns the payoff $\pi_{it}$.  
 We normalize the interest rate of Bank $L$ to $1$, so that its profit $\Pi$ in period $t$ is given by 

    $$\Pi_{Lt} = \sum_i S_{iLt}A_{iLt}\pi_{it}$$
Bank $H$ charges an interest rate of $1+\nu_t$. It has profit function:
$$
    \Pi_{Ht} = \sum_i S_{iHt}A_{iHt}^{(t)}\left[ (1+\nu_t)\max\left\{\pi_{it},0\right\}+\min\left\{\pi_{it},0\right\}\right].
$$
Here, positive returns are amplified by a factor of $1+\nu_t$, whereas losses are incurred on a dollar-for-dollar basis. 
\subsubsection{Risk Management Constraints}
\label{sec:constraints}
We next suppose that banks face a \textit{risk management} or \textit{solvency} constraint. That is, there is some level of financial loss that is `unacceptable' to the bank. This may be due to regulatory constraints, such as Basel III or Dodd-Frank, or liquidity constraints \citep{VaR_Basel}. This loss threshold is probabilistic: the bank is willing to accept some nonzero risk that their loss from making a loan falls below a certain threshold, but wants to specifically limit the probability that this occurs to less than a risk tolerance $\alpha\%$. Typically, this value is set to $1\%$ or 
$5\%$ \citep{Jorion2000-ng}.

We model this with a per-period Value-at-Risk ($\VaR$) constraint \citep{Artzner1999}. That is, we require that the bank's anticipated profit in each period must be greater than some constant $\rho < 0$ with probability at least $1-\alpha$. We can think of $\rho$ as the bank's maximum acceptable loss. Equivalently, the bank is willing to accept an $\alpha\%$ risk that their profit will fall below the bank's risk management threshold $\rho$.    
\begin{defn}[Value-at-Risk ($\VaR$)]
$$VaR_\alpha(X) = - \inf\left\{x \,|\, \Prob\left[X \leq x \right] > \alpha \right\}$$ 
\end{defn}

We adapt this by requiring that each bank $j$ in period $t$ faces the constraint that: 
    $$ \Prob\!\left( \Pi_{jt} < \rho \right) \le \alpha$$
This simply states that the bank will accept a risk of at most $\alpha\%$ that their profit in period $t$ falls below $\rho$.
 
We suppose that this constraint is \textit{lexically prior} to the profit-maximization objective: the bank must satisfy the VaR constraint in order to lend at all. 

 Combining each bank's payoff function with their VaR constraint, we can write each bank's optimization problem as follows: 
 
    $$\underset{\bm{A}_{jt}}{\arg\max}\;\; \Pi_{jt}\left(\bm{A}_{jt}\right) \quad \text{subject to} \quad \Prob\!\left( \Pi_{jt}\left(\bm{A}_{jt}\right) \leq \rho \right) \le \alpha$$

This says simply that each banks chooses the approval decisions that maximize profits subject to satisfying their VaR constraint.  As we shall see below, the maximand depends on the expected value of the applicants' payoff function $\mu_i$, while the VaR depends on both the mean and the variance of the applicant group's payoff function. 

\subsubsection{Prior Bank Beliefs}

The actual profit $\Pi_{jt}$ is unknown \textit{ex ante}, however: it is only observed once a loan has been approved and accepted. Banks must therefore instead make decisions based on prior beliefs about the profitability of approving a loan. 

Note that, at $t = 1$, the bank cannot use observed repayment information in the determination of the loan application. Instead, we suppose that the bank uses historical credit data to estimate the risk of lending to specific groups. 

Denote by $\hat{\Pi}_{jt}$ the bank's posterior estimate of profit from lending at time $t$. Each bank instead must solve the \textit{feasible} problem:

   $$\underset{\bm{A}_{jt}}{\arg\max}\;\; \hat{\Pi}_{jt}\left(\bm{A}_{jt}\right) \quad \text{s. t.} \quad \Prob\!\left( \hat{\Pi}_{jt}\left(\bm{A}_{jt}\right) \leq \rho \right) \le \alpha$$

\noindent Which, importantly, depends on bank beliefs about applicants. 

We assume that each bank $j$ has initial prior beliefs regarding the parameters of the distribution of returns from issuing a loan to a member of group $i$. 

We assume, as before, that the mean of both payoff distributions is known and common knowledge. We assume that the scale of each distribution is unknown, however. 

\begin{defn}[Bank's prior beliefs]
We suppose that each bank has an Inverse-Gamma prior on the variance of each group, so that:
$$\sigma^2_{ij,0} \sim \text{Inv-}\Gamma(a_0, b_0)$$
\end{defn}

\subsubsection{Differential Credit Records Drive Prior Variances}

We suppose that each bank estimates a model of repayment probability based on historical credit data to form its initial beliefs. That is, on having observed a credit file with effective sample size $C_i = np_i$, each bank forms a prior about each group's variance as follows:
$$\sigma^2_{ij} | C_i \sim \text{Inv-}\Gamma\left(a_0 + \dfrac{np_i}{2},b_0 + \dfrac{np_iS^2_i}{2}\right)$$

Since we suppose that $p_B < p_W = 1$, we have (assuming that $S^2_B = S^2_W = 1$):
\begin{align*}
\sigma^2_{Bj,0} | C_B &\sim \text{Inv-}\Gamma\left(a_0 + \dfrac{np_B}{2},b_0 + \dfrac{np_B}{2}\right) \\ 
\sigma^2_{Wj,0} | C_W &\sim \text{Inv-}\Gamma\left(a_0 + \dfrac{n}{2},b_0 + \dfrac{n}{2}\right)
\end{align*}
So that:
\[\hat{\sigma}^2_{Bj,0} = \mathbb{E}[\sigma^2_{Bj,0} | C_B]  >  \mathbb{E}[\sigma^2_{Wj,0} | C_W] = \hat{\sigma}^2_{Wj,0}\]

Or, in other words, based on observable credit histories, each bank has a prior belief that the variance of the minority group is larger than the variance of the majority group. 

\subsubsection{Learning Through Lending}\label{sec:updating}
In each period, if a bank extends a loan to an applicant from group $i$, it observes the realized repayment and updates its risk assessment for that group via Bayesian updating. 

\begin{lem}[Belief Updating via Bayes Rule]
If bank $j$ lends to group $i$, it observes a return $\pi_{it}$ and then updates its posterior beliefs about the variance $\sigma_i^2$. Suppose that the bank observes a sequence of returns $\{\pi_1 ,\pi_2, \dots, \pi_M\}$. Then, it has posterior belief about the variance of group $i$: 
\begin{align*}
&\sigma_B^2 | \pi_1, \dots, \pi_M \sim  \\
&\text{Inv-}\Gamma \left(a_0 + \dfrac{np_B + M}{2}, b_0 + \dfrac{np_B + \sum_{m=1}^M(\pi_{Bm} - \mu_B
)^2}{2} \right)
\end{align*}
So that:
\[\mathbb{E}[\sigma_B^2 | \pi_1, \dots, \pi_M ] = \dfrac{b_0 + \dfrac{np_B + \sum_{m=1}^M(\pi_{Bm} - \mu_B
)^2}{2}}{a_0 + \dfrac{np_B + M}{2} - 1}\]
\end{lem}

In other words, the bank's posterior belief about the variance of the minority group depends both on the depth of historical credit files ($np_B$) and the observed variance of repayments ($\sum_{m=1}^M(\pi_{Bm} - \mu_B
)^2$), conditional on lending to the minority group. 

A key assumption in our analysis is that a bank updates its estimate of the variance for group $i$ only when it observes repayment outcomes from that group. Thus, if a bank never lends to group $B$, it will not receive data on the performance of loans to that group, and incorrect prior belief about $B$'s variance will not be updated.

\subsection{Our Modeling Choices}

\noindent\textbf{VaR and CVaR/ES} Value-at-Risk is a commonly used risk management measure in practice, and is explicitly referenced in the Basel III banking regulations \citep{VaR_Basel, CHANG2019}. We could also study the Conditional Value-at-Risk, or Expected Shortfall, which is related to VaR via the formula: $ES(X) = \frac{1}{\alpha}\int_0^{\alpha}VaR_\phi
(X)\: d\phi$. The intuitions underlying our model carry through to the CVaR/ES risk metric: there is a threshold variance at which the CVaR/ES will be violated. We show that the CVaR/ES threshold is more conservative than the VaR risk metric, and hence requires a larger subsidy to escape the subprime trap, in the Appendix. \\

\noindent\textbf{Inaccurate prior beliefs} Our model supposes that banks have inaccurate beliefs about the creditworthiness of applicants. This is a plausible outcome when decisions are made on the basis of credit scoring data that is either biased or noisy. Credit scores may be inherently noisier representations of underlying default risk for minority groups \citep{blattner2021how}, while the components of the score may themselves be low quality indicators of repayment ability \cite{Rice2013DiscriminatoryEO}. Empirical work shows that banks are more likely to avoid lending to minority neighborhoods, thereby lowering the quality of information received about lenders in those neighborhoods \cite{blattner2021costlynoisedatadisparities}.  Noisier estimates of credit risk are likely to lead to inflated prior estimates of variance in practice. \\

\noindent\textbf{Risk-pooling} We study three cases: where banks lend unilaterally to each group, and where banks lend to both groups. In the unilateral setting, there is no risk pooling, and banks make decisions based only on single group characteristics. However, when banks are willing to lend to both groups, risk pooling occurs, where the lower-risk group essentially subsidizes the risk of lending to the higher-risk group. This case is studied in \citep{donahue_barocas2021}. Below, we derive variance thresholds for each of the three cases.

\section{Single-Period Setting}
\label{sec:single-period}

We first study the stage game, to show how the VaR constraint affects lending to group $B$ in each period.
Note first that we assume that Bank H charges an interest rate premium to derive variance thresholds. We then prove that this occurs. 
Second, note that we derive variance thresholds for the pooled variance -- thresholds under which banks will lend to both groups. Because the variance thresholds are ordered, it follows that if the bank's variance pooled threshold is not satisfied, it will not unilaterally lend to the minority group. It also follows that if the bank's pooled threshold is satisfied, the bank will unilaterally lend to the majority group. 

First we assume that Banks have unilateral thresholds for lending to each group: 

\begin{lem}[Unilateral Lending Thresholds]
\begin{align*}
   \tilde{\sigma}^2_L &= \left(\dfrac{\rho_L - \mu}{\Phi^{-1}(\alpha)}\right)^2 \\
    \tilde{\sigma}^2_H &=\left(\dfrac{\rho_L - (1+\nu)\mu}{\Phi^{-1}(\alpha)}\right)^2
\end{align*}
\end{lem}

These represent the thresholds under which each bank is willing to lend unilaterally to a given group.

Each bank also has a pooled variance threshold, under which it is willing to lend to both groups. 

\begin{lem}[Pooled Variance Threshold for Bank $L$]
\label{lem:SigmaLThreshold}
Consider Bank $L$, which offers loans at a normalized interest rate of $1$, lending to group $B$ with a random payoff given by $\pi_i\sim \N(\mu_i,\sigma_i^2)$. 
Bank $L$ lends to both groups in period $t$ if and only if group B's variance satisfies:
$$\hat{\sigma}^2_{BLt} \leq  \tilde{\sigma}_L^{2, \textrm{pool}} \equiv \left(\frac{\rho_L - (\mu_W + \mu_B + \Phi^{-1}(\alpha)\sigma_W)}{\Phi^{-1}(\alpha)}\right)^2$$
\end{lem}

In other words, there is an upper bound $\tilde{\sigma}^L$ on the Bank's beliefs about the variance of group B, such that the bank is only willing to lend to group $B$ if it believes that group has lower variance than this threshold. Otherwise, its risk of a shortfall exceeds $\alpha\%$. 

A similar analysis for the subprime bank $H$ shows that it tolerates a higher level of variance. We have:

\begin{lem}[Pooled Variance Threshold for Bank $H$]
\label{lem:SigmaLThreshold}
Bank $H$ lends to group $B$ and group $W$ in period $t$ if and only if $\hat{\sigma}^2_{BHt} \leq  \tilde{\sigma}^H_{\text{pool}}$, where: 
$$\tilde{\sigma}_H^{2, \text{pool}} = \left(\frac{\rho_H - \left[(1+ \nu_t)(\mu_W + \mu_B ) + \Phi^{-1}(\alpha)\sigma_W\right]}{\Phi^{-1}(\alpha)}\right)^2
$$
\end{lem}
\begin{corollary}[Ordered variance thresholds]
While $\nu_t > 0$, and for $\alpha < .1$: 
    $$0 < \tilde{\sigma}^2_L < \tilde{\sigma}^{2, \text{pool}}_L < \tilde{\sigma}^2_H <  \tilde{\sigma}_H^{2, \text{pool}}$$
\end{corollary}

Intuitively, gains are scaled by the factor $(1+\nu_t)$, providing the high-interest rate bank  insulation against downside risk. Thus, when the perceived variance of group $i$ is between $\tilde{\sigma}^2_L$ and $\tilde{\sigma}^2_H$, the low-interest Bank $L$ will decline to lend, while the high-interest Bank $H$ can approve the loan.

\section{The Subprime Trap}
\label{sec:trap}

We now consider the multi-period setting in which lending decisions influence the evolution of banks' beliefs about borrower creditworthiness. 
\subsection{Assumptions}
We make several assumptions. These model stylized facts that characterize the subprime trap. First, we suppose that the variance of both groups is \textit{in fact} below the threshold $\tilde{\sigma}^2_L$, though the variances are not necessarily identical. 
\begin{asm}[Both groups are in fact creditworthy]
$$\sigma^2_W \leq \sigma^2_B < \tilde{\sigma}^2_L$$
\end{asm}
This is intended to describe a situation in which the low-interest rate bank would lend to both groups under perfect information. 

We characterize beliefs under imperfect information in which lending to group B does not occur. We suppose that Bank $L$ holds a prior that the variance of repayments for group $B$ is $\hat{\sigma}^2_{BL,0}$ such that:
\begin{asm}[L's prior variance for group B is above its risk threshold]
$\hat{\sigma}^2_{BL,0} > \tilde{\sigma}^{2, \text{pool}}_L$
\end{asm}
Where $\tilde{\sigma}^{2, \text{pool}}_L$ is the maximum variance that Bank $L$ (the main lender) can tolerate under its constraint $\VaR$. Based on its initial assessment, lending violates the risk limits of Bank $L$. 

Third, Bank $H$, which operates at a higher interest rate and has a correspondingly higher risk tolerance. The high-rate bank has a lending threshold $\tilde{\sigma}^2_H$ such that $\tilde{\sigma}^2_H > \tilde{\sigma}^{2,\text{pool}}_L$. Denoting bank H's prior by $\hat{\sigma}^2_{BH, 0}$, we suppose that: 

\begin{asm}[H's prior variance for group B is below its risk threshold]
$\hat{\sigma}^2_{BH, 0}\leq \tilde{\sigma}^2_H$
\end{asm}

Consequently, Bank $H$ is prepared to lend to group $B$. 
\subsection{Equilibrium}
In each period, banks choose simultaneously whether to lend to groups $W$ and $B$. Both types of bank approve group $W$'s loan application, and group $W$ chooses the lower-rate, mainstream bank. 
Based on their prior beliefs about B's variances, $H$ approves $B$'s loan application, but $L$ does not. If Bank $L$ does not lend to group $B$, then group $B$ is left with the subprime option from Bank $H$. But since Bank $L$'s belief update depends on observing a return from group $B$, if it does not lend to group $B$ at period $t$, it does not have an updated posterior belief to use as its prior at  period $t+1$: its beliefs do not change. 
The crucial observation is that if Bank $L$ persistently withholds loans to group $B$, it never observes the repayment data needed to update its inflated variance estimate, and its belief remains at $\hat{\sigma}_{BL, 0}^{2}$. 

We now formalize this result. 

\begin{thm}[Subprime Trap Equilibrium]
\label{thm:SubprimeTrap}
Suppose Assumptions 1-5 hold. Then, there exists a Bayesian subgame-perfect equilibrium in which, in every period, Bank $L$ lends exclusively to group $W$, and Bank $H$ lends to group $B$ with high probability, permanently relegating group $B$ to subprime loans.
\end{thm}

\begin{proof}
We show by induction that the low-cost bank never lends to group $B$.

First, in the base case, the bank does not lend to group B. We have that $A_{ij1} = 0$, since, by Lemma 1, if the bank's prior belief is $\hat{\sigma}^2_{BL0}$ and $\hat{\sigma}^2_{BL0}>\tilde{\sigma}^{2, \text{pool}}_L$, so that the VaR constraint is violated, by Lemma 1, and hence lending does not occur in period 1. 

To show the induction step, we show that if the bank does not lend to B in period $t$, then it does not lend to B in period $t+1$. For this, we have that $A_{ijt} = 0 \implies \hat{\sigma}^2_{ijt+1} =   \hat{\sigma}^2_{ijt}$, because no updating occurs. but since $A_{ijt+1} = 1$ if and only if $\hat{\sigma}_{ijt+1} \leq \tilde{\sigma}^{2,\text{pool}}_L$, we have that $\hat{\sigma}^2_{Bt} > \tilde{\sigma}^{2, \text{pool}}_L$  and $\hat{\sigma}^2_{ijt+1} =   \hat{\sigma}_{ijt}^{2}  \implies \hat{\sigma}^2_{ijt+1} > \tilde{\sigma}^2_L$, so that $A_{ijt+1} = 0$. This verifies the induction step. 
 
This shows by induction that the low-cost bank never lends to $B$. Since the bank \textit{does} lend to $W$, we have that the optimal decision for Bank $L$ is to lend exclusively to group $W$ in every period.

Next, consider Bank $H$. By hypothesis, Bank $H$ believes that $\sigma^2_{BH0}\le\tilde{\sigma}^2_H$, and therefore lending to group $B$ satisfies its risk constraint in the initial period. Moreover, since the true mean $\mu_B$ is positive, lending to group $B$ is expected to yield a positive return. Bank $H$ lends to group $B$ with high probability, because in each stage there is a (vanishing) probability of observing a payoff realization that increases Bank $H$'s posterior belief about group B's variance above its threshold. This probability goes to zero by the SLLN, however. 

The borrowers act accordingly. Group $W$, faced with a lower interest rate from Bank $L$, accepts that offer, while group $B$, having no offer from Bank $L$, accepts the loan from Bank $H$.

Finally, note that because Bank $L$ never lends to group $B$, it never observes any repayment outcomes from that group. This implies that the variance estimate for group $B$ remains at $\hat{\sigma}^2_{BL0}>\tilde{\sigma}^{2, \text{pool}}_L$ in every period. As a consequence, there is no incentive for Bank $L$ to deviate from its strategy of lending only to group $W$, and similarly, neither Bank $H$ nor the borrower groups have incentives to deviate from their prescribed actions. Hence, the described strategy profile constitutes a subgame-perfect equilibrium.
\end{proof}

\begin{remark}[Purely Informational Failure]
Even though the true variance of group $B$ is such that $\sigma^2_B < \tilde{\sigma}^{2, \text{pool}}_L$, Bank $L$ remains unaware of this fact unless it extends credit to group $B$. The absence of new data prevents Bank $L$ from updating its risk assessment, leading to an inefficient outcome in which group $B$ continues to access only subprime credit.
\end{remark}

\begin{remark}[Equal Creditworthiness]
It is noteworthy that the equilibrium outcome emerges solely from differences in prior variances. Since $\mu_B = \mu_W = \mu$, the two groups are identical in terms of average creditworthiness; the discrepancy stems entirely from the combination of incorrect beliefs about borrower variance, and the risk assessment constraints imposed by $\VaR$.
\end{remark}

\section{Escaping the Subprime Trap via Subsidies}
\label{sec:subsidies}

We now show that a subsidy or partial lending guarantee can help group $B$ escape the subprime trap. The key idea is to cover sufficient downside risk so that Bank $L$ is induced to lend to group $B$. When this occurs, Bank $L$ can gather information on group $B$'s repayment performance. Since, by Assumption 1, group $B$, is in fact creditworthy, Bank $L$'s beliefs will eventually converge to group $B$'s actual variance, which, by Assumption 2, is low enough that $L$'s VaR constraint will be satisfied. 
Once enough observations have been accumulated, Bank $L$ will update its risk assessment so that the estimated variance satisfies $\hat{\sigma}^2_{BLt}<\tilde{\sigma}^2_L$, eliminating the need for further subsidy.

\begin{lem}[Learning through Lending]
 ${{\hat{\sigma}^2_{BLm}}}\to\sigma^2_B$ almost surely as $m\to\infty$
\end{lem}
By Lemma 3, this is the condition required for Bank L to lend to group $B$, so that the $\VaR$ constraint is satisfied. In other words, if the mainstream bank \textit{were} to lend to group $B$ for a sufficiently long period of time, it would learn about the variance of returns due to lending to group $B$. The problem is then how to induce the bank to do so. 

To encourage Bank $L$ to lend to group $B$, we introduce a subsidy mechanism. 

First, we define our subsidy as the smallest side-payment that would allow the bank to satisfy its VaR constraint in period $t$. We have:

$${s_{t}^*} = \inf\left\{ s \geq 0 \: : \: \Prob\left(\Pi_{jt} + s < \rho\right) \leq \alpha\right\}$$

This allows us to solve for the optimal subsidy.
\begin{lem}
In each period $t$, the required subsidy is equal to:
$${s_t^*} =  \max\Bigl\{0,\;\rho-(\mu_W + \mu_B)-\Phi^{-1}(\alpha)\,(\hat{\sigma}_{BLt} + \sigma_W)\Bigr\}$$
\end{lem}

If $s_t^*\le0$, no subsidies are needed; otherwise, the regulator covers any shortfall up to $s_t^*$, ensuring that Bank $L$ meets its $\VaR$ restriction.

We formalize the adaptive subsidy mechanism in Algorithm~\ref{alg:AdaptiveSubsidy}.

\begin{algorithm}[H]
\caption{Adaptive Subsidy to Escape the Subprime Trap}
\label{alg:AdaptiveSubsidy}
\begin{algorithmic}[1]
\Require Parameters: $\rho$, $\alpha$, $\mu$, prior parameters $(a_0, b_0)$, $np_B$, threshold $\tilde{\sigma}^2_L$, horizon $T$
\State \textbf{Initialize:}
\State $a_t \gets a_0 + \frac{np_B}{2}$ \Comment{Shape with historical data}
\State $b_t \gets b_0 + \frac{np_B}{2}$ \Comment{Scale with historical data}
\State $M \gets 0$ \Comment{Count of new loans to group B}
\For{$t \in 1,\cdots, T$}
    \State $\hat{\sigma}^2_{BLt} \gets \frac{b_t}{a_t - 1}$ \Comment{Current estimate}
    \If{$\hat{\sigma}^2_{BLt} > \tilde{\sigma}^*_L$}
        \State $s_t \gets s^*_t(\hat{\sigma}^2_{BLt})$
        \State Bank L lends to B with subsidy $s_t$
        \State Observe return $\pi_{B,t}$
        \State $M \gets M + 1$
        \State $a_t \gets a_0 + \frac{np_B + M}{2}$
        \State $b_t \gets b_0 + \frac{np_B + \sum_{m=1}^{M}(\pi_{B,m} - \mu)^2}{2}$
    \Else
        \State No subsidy needed. 
    \EndIf
\EndFor
\end{algorithmic}
\end{algorithm}
The subsidy incentivizes Bank $L$ to lend to group $B$, generating observations that allow its risk assessment to converge to the true variance $\sigma_B^2$. Once the updated estimate satisfies $\hat{\sigma}^2_{BLt}<\tilde{\sigma}^2_L$, it follows from Lemma 2 that Bank $L$ meets its $\VaR$ constraint without subsidy (i.e., ${s_t^*}$), and no further external support is required.

We now state the main result regarding the subsidy mechanism.

\begin{thm}[Escaping the Subprime Trap]
\label{thm:SubsidyWorks}
Assume that the true distribution of repayments for group $B$ is $\N(\mu,\sigma_B^2)$ with $\sigma^2_B<\tilde{\sigma}^2_L$, but that Bank $L$ initially believes that the variance is $\hat{\sigma}^2_{BL0}>\tilde{\sigma}^2_L$. Under the subsidy mechanism described above (Algorithm~\ref{alg:AdaptiveSubsidy}), with probability one there exists a finite time $\tau$ such that for all $t\ge\tau$, Bank $L$ updates its variance estimate so that $\hat{\sigma}^2_{BLt>\tau} \to \sigma^2_B <\tilde{\sigma}^L$ and consequently ${s_t^*} \to 0$. Beyond time $\tau$, Bank $L$ continues lending to group $B$ without further subsidy.
\end{thm}

In short, the subsidy induces the bank to update their information about minority group creditworthiness. It is also temporary when there exist otherwise creditworthy borrowers trapped in subprime loans: subsidizing information discovery helps banks to learn the true variance of lender groups. Further, while algorithmic audits could correct bias, subsidies can avoid privacy hurdles and scale faster \cite{Cowgill2019}.

\begin{corollary}[Long-term benefits of subsidy]
For $t > \tau,\: \nu_t = 0$.
\end{corollary}

In any round in which the low-interest bank would approve group $B$'s loan, the high-interest bank will choose to set its premium to zero. If it chooses any positive premium, neither group will accept the loan. If it chooses a premium of zero, both groups will randomize over accepting the loan from bank $B$ and accepting the loan from bank $H$. Since this persists for all $t > \tau$, this reduces interest rates permanently: competition reduces the premiums charged by subprime lenders. 

\section{Conclusion}

Inequalities in algorithmic lending practices can persist even when minority applicants are as creditworthy as white applicants. We study one explanation for this phenomenon: the role of risk-management constraints, specifically $\VaR$, in contributing to persistent disparities in lending. We have shown that risk management constraints can lead lenders to refuse loans even with positive Net Present Value. This forces minority applicants to accept high-interest rate loans, with negative implications for their financial prospects, a situation we describe as the ``subprime trap" equilibrium. We emphasize the role of inaccurate prior beliefs about the risk of lending to minority groups, which leads to systematic exclusion and higher borrowing costs. 

Our theoretical results align with documented patterns in mortgage lending markets.  First, the persistence of our subprime trap equilibrium mirrors the decade-long stability of racial lending disparities documented in mortgage markets, even as algorithmic lending has expanded \citep{Bartlett2022}. Second, our model's prediction that mainstream banks never update their beliefs about minority creditworthiness is consistent with \citep{Bayer2016}. Third, the equilibrium's existence despite equal creditworthiness  directly explains the puzzle identified by \citep{Popick2022} and \citep{Crosignani2023}: why minorities with comparable credit scores remain concentrated at high-cost lenders.

Implementation of our subsidy mechanism maps naturally onto existing policy infrastructure. The Federal Housing Administration's mortgage insurance programs already function as partial guarantees, covering lender losses when borrowers default on FHA-insured loans \cite{hud2024sfh}. Similarly, the Community Reinvestment Act provides regulatory credits to banks for community development loans and investments in low- and moderate-income neighborhoods \cite{cra1977, fedreserve2024cra}. Our analysis suggests these programs could achieve greater impact if restructured to explicitly target variance reduction during a learning period; guaranteeing tail risks for a finite horizon would enable banks to update their risk assessments while avoiding solvency-affecting losses.

Our analysis focuses on one specific mechanism: variance-based discrimination through VaR constraints. Observed lending disparities likely result from multiple interacting factors including historical discrimination, geographic segregation, and differential access to financial services.

Nonetheless, targeted, temporary interventions, such as subsidies or guarantees, can break this cycle by allowing banks to learn true risk profiles. These findings suggest practical avenues for addressing lending disparities and offer a framework to explain failures of fairness in algorithmic lending decisions. 

\bibliography{refs_lending}

\appendix
\section{Appendix}
\subsection{Proof of Lemma 1}
\begin{proof}
We derive the posterior distribution for $\sigma^2_B$ after observing $M$ loan outcomes.

\noindent \textbf{Initial beliefs based on historical credit data.}

\noindent Before observing any data, the bank starts with an uninformative prior:
$$\sigma^2_B \sim \text{Inv-}\Gamma(a_0, b_0)$$

The bank observes $np_B$ historical loan outcomes from group B applicants with complete files. 

The likelihood of observing the historical data given $\sigma^2_B$ is:
$$p(C_B | \sigma^2_B) \propto (\sigma^2_B)^{-np_B/2} \exp\left(-\frac{np_B}{2\sigma^2_B}\right)$$

By Bayes' theorem:
$$p(\sigma^2_B | C_B) \propto p(C_B | \sigma^2_B) \cdot p(\sigma^2_B)$$

$$\propto (\sigma^2_B)^{-np_B/2} \exp\left(-\frac{np_B}{2\sigma^2_B}\right) \cdot (\sigma^2_B)^{-a_0-1} \exp\left(-\frac{b_0}{\sigma^2_B}\right)$$

$$= (\sigma^2_B)^{-(a_0 + \frac{np_B}{2})-1} \exp\left(-\frac{b_0 + \frac{np_B}{2}}{\sigma^2_B}\right)$$

Therefore, the initial belief is:
$$\sigma^2_{B,0} \sim \text{Inv-}\Gamma\left(a_0 + \frac{np_B}{2}, b_0 + \frac{np_B}{2}\right)$$

\noindent\textbf{Learning through lending}.

\noindent Given the true variance $\sigma^2_B$ and known mean $\mu_B$, the observed returns $\{\pi_{B,1}, \ldots, \pi_{B,M}\}$ are independently distributed:
$$\pi_{B,m} | \sigma^2_B \sim \mathcal{N}(\mu_B, \sigma^2_B) \quad \text{for } m = 1, \ldots, M$$

The prior is:
$$p(\sigma^2_B) \propto (\sigma^2_B)^{-{a_0 + \frac{np_B}{2}}-1} \exp\left(-\frac{b_0 + \frac{np_B}{2}}{\sigma^2_B}\right)$$

The likelihood for $M$ observations with known mean $\mu_B$ is:
$$p(\pi_{B,1}, \ldots, \pi_{B,M} | \sigma^2_B) = \prod_{m=1}^M \frac{1}{\sqrt{2\pi\sigma^2_B}} \exp\left(-\frac{(\pi_{B,m} - \mu_B)^2}{2\sigma^2_B}\right)$$

$$= (2\pi\sigma^2_B)^{-M/2} \exp\left(-\frac{\sum_{m=1}^M(\pi_{B,m} - \mu_B)^2}{2\sigma^2_B}\right)$$

$$\propto (\sigma^2_B)^{-M/2} \exp\left(-\frac{\sum_{m=1}^M(\pi_{B,m} - \mu_B)^2}{2\sigma^2_B}\right)$$

By Bayes' theorem:
$$p(\sigma^2_B | \text{data}) \propto p(\text{data} | \sigma^2_B) \cdot p(\sigma^2_B)$$

$$\propto (\sigma^2_B)^{-M/2} \exp\left(-\frac{\sum_{m=1}^M(\pi_{B,m} - \mu_B)^2}{2\sigma^2_B}\right) \cdot (\sigma^2_B)^{-a_0-\frac{np_B}{2}-1} \exp\left(-\frac{b_0 + \frac{np_B}{2}}{\sigma^2_B}\right)$$

$$= (\sigma^2_B)^{-(a_0+\frac{np_B+M}{2})-1} \exp\left(-\frac{b_0 + \frac{np_B + \sum_{m=1}^M(\pi_{B,m} - \mu_B)^2}{2}}{\sigma^2_B}\right)$$

This is the kernel of an Inverse-Gamma distribution with parameters:
$$a_{post} = a_0 + \frac{np_B + M}{2}, \quad b_{post} = b_0 + \frac{np_B + \sum_{m=1}^M(\pi_{B,m} - \mu_B)^2}{2}$$

so that: 
$$\sigma^2_B | \pi_{B,1}, \ldots, \pi_{B,M} \sim \text{Inv-}\Gamma\left(a_0 + \frac{np_B + M}{2}, b_0 + \frac{np_B + \sum_{m=1}^M(\pi_{B,m} - \mu_B)^2}{2}\right)$$
The mean of an Inverse-Gamma$(a,b)$ distribution is $\frac{b}{a-1}$. Therefore:
$$\mathbb{E}[\sigma^2_B | \pi_{B,1}, \ldots, \pi_{B,M}] = \frac{b_0 + \frac{np_B + \sum_{m=1}^M(\pi_{B,m} - \mu_B)^2}{2}}{a_0 + \frac{np_B + M}{2} - 1}$$

Which is the formula for the bank's posterior estimate of the variance of lending to group $B$, given that it has observed M loans in addition to prior group-level credit history data. 
\end{proof}
\subsection{Proof of Lemma 2}

\begin{proof}
First, we suppose that Bank $L$ lends to group $W$, and derive a threshold rule for the prior variance $\sigma^2_B$ holding $\sigma^2_W$ fixed such that this rule determines whether or not Bank $L$ will lend to group $B$.

First, we note that it is profitable in expectation for Bank $L$ to lend to group $B$, because, by Assumption 1, $\mathbb{E}[\pi_B] = \mu_B > 0$. 

We then need to check Bank $L$'s $VaR$ constraint. It is:

$$Pr\left(\Pi_{jt} < \rho \;|\; A_{BLt} = 1   \right) = \Phi\left(\dfrac{\rho_L - (\mu_B + \mu_W)}{\hat{\sigma}_{Bt} + \sigma_W}\right)$$

Where $\Phi(\cdot)$ is the Cumulative Distribution Function of the standard normal distribution.
This entails that 
$$Pr\left(\Pi_{jt} \, |\, A_{BL} = 1  < \rho \right) \leq \alpha \: \Longleftrightarrow \Phi\left(\dfrac{\rho_L - (\mu_B + \mu_W)}{\hat{\sigma}_{Bt} + \sigma_W}\right) \leq \alpha$$

Rearranging, we get the condition that:

$$\hat{\sigma}^2_{Bt} \leq \left(\dfrac{\rho_L - (\mu_W + \mu_B + \Phi^{-1}(\alpha)\sigma_W)}{\Phi^{-1}(\alpha)}\right)^2$$ 

Where the right-hand side defines a variance threshold $\tilde{\sigma}^2_L$ such that, if group B's prior variance is above this threshold, Bank $L$ will not lend to group $B$:

$$\hat{\sigma}^2_{BLt} \leq \tilde{\sigma}^{2, \text{pool}}_L \equiv  \left(\dfrac{\rho_L - (\mu_W + \mu_B + \Phi^{-1}(\alpha)\sigma_W)}{\Phi^{-1}(\alpha)}\right)^2$$ 

Next, we show that Bank $L$ always lends to group $W$. Recall that, by Assumption 2, Group $W$'s variance is common knowledge: this reflects the assumption that banks have higher quality information about the creditworthiness of the majority group.

Checking the VaR constraint in the unilateral case, we generate the threshold:

$$\sigma^2_W \leq \left(\dfrac{\rho_L - \mu_W}{\Phi^{-1}(\alpha)}\right)^2$$

Since we know that, by Assumption 3, this condition is satisfied, Bank $L$ is willing to lend to $W$ unilaterally. Assumption 3 also tells us that $\hat{\sigma}^2_{BL0}$ is above this threshold, indicating that Bank $L$ will not lend to $B$ unilaterally. 

Hence, Bank $L$ will lend to group $W$ unilaterally, and will lend to group $B$ in addition to group $W$ if and only if group $B$'s variance satisfies $\hat{\sigma}^{2}_{BLt} \leq \tilde{\sigma}^2_L$.
\end{proof}
\subsection{Proof of Lemma 3}
\begin{proof}
Bank $H$'s VaR constraint can be written as:

$$Pr\left(\Pi_{jt} < \rho | A_{BL} = 1   \right) = \Phi\left(\dfrac{\rho_L - [1 + \nu_t](\mu_B + \mu_W)}{\hat{\sigma}_{Bt} + \sigma_W}\right)$$

Rearranging, we get:

$$\hat{\sigma}_{BHt}^{2} \leq  \tilde{\sigma}^{2, \text{pool}}_H \equiv \left(\dfrac{\rho_L - [(1+\nu_t)(\mu_W + \mu_B) + \Phi^{-1}(\alpha)\sigma_W]}{\Phi^{-1}(\alpha)}\right)^2$$
\end{proof}

\subsection{Proof of Corollary 1}
\begin{proof} This can be deduced by inspection, but note that, by convention, $\alpha < .1$, so that $\Phi^{-1}(\alpha) < 0$. \end{proof}

\subsection{Proof of Lemma 4}
\begin{proof}This follows by application of the Law of Large Numbers. \end{proof}

\subsection{Proof of Lemma 5}

\begin{proof}
In each period, Bank L has belief, $\hat{\pi}_{Bt}\sim\N\bigl(\mu,\hat{\sigma}^2_{BLt}\bigr)$, so that: 
\[
\Prob\!\left(\Pi_{jt}+s<\rho_L\right)
=\Phi\!\left(\frac{\rho_L-(\mu_W + \mu_B) -s}{\hat{\sigma}_{BLt} + \sigma_W}\right).
\]
Thus, to guarantee $\Phi\!\left(\frac{\rho_L-(\mu_W + \mu_B) -s}{\hat{\sigma}^2_{BLt} + \sigma_W}\right)\le\alpha$, we require
\[
\frac{\rho_L-(\mu_W + \mu_B)-s}{\hat{\sigma}_{BLt} + \sigma_W}\le\Phi^{-1}(\alpha),
\]
which rearranges to
\[
s\ge\rho_L-(\mu_W + \mu_B)-\Phi^{-1}(\alpha)\,(\hat{\sigma}_{BLt} + \sigma_W).
\]
The minimal such subsidy is therefore:
\[
s_t^*=\max\Bigl\{0,\;\rho_L-(\mu_W + \mu_B)-\Phi^{-1}(\alpha)\,(\hat{\sigma}_{BLt} + \sigma_W)\Bigr\}.
\]
\end{proof}

\subsection{Proof of Theorem 2}
\begin{proof}

We prove the theorem in three steps.

First, by design the subsidy ${s^*}^{(t)}$ ensures that when Bank $L$ lends to group $B$, the bank's $\VaR$ constraint is satisfied, that is,
\[
\Prob\!\:\Bigl(\Pi_{jt}+{s_t^*}<\rho\Bigr)\le\alpha.
\]
Thus, the presence of ${s_t^*}$ makes lending feasible even under Bank $L$'s initial risk assessment $\hat{\sigma}^2_{BL0}>\tilde{\sigma}^L$. Consequently, if the expected net profit (inclusive of the subsidy) is positive, Bank $L$ has an incentive to lend to group $B$.

Second, every time Bank $L$ extends a loan to group $B$, it observes a repayment drawn from $\N(\mu,\sigma_B^2)$. Let $m$ be the number of such observations. By the Strong Law of Large Numbers, 
\[
\hat{\sigma}_{BLm}^2 \to \sigma_B^2 \quad \text{almost surely as } m\to\infty.
\]
Because $\sigma^2_B<\tilde{\sigma}^{2, \text{pool}}_L$, there exists a (finite) index $m^*$ (and thus a finite time $\tau$) such that for all $m \ge m^*$ the updated estimate satisfies:
\[
\sigma^2_B\le \hat{\sigma}^2_{BLm}  < \tilde{\sigma}^{2, \text{pool}}_L.
\]

Third, once the updated variance estimate satisfies $\hat{\sigma}^2_{BLt}<\tilde{\sigma}^{2, \text{pool}}_L$, we can evaluate the required subsidy as
\[
{s^*}^{(t)}=\max\Bigl\{0,\rho_L-(\mu_W + \mu_B)-\Phi^{-1}(\alpha)\,(\hat{\sigma}_{BLt} + \sigma_W)\Bigr\}=0,
\]
by the same calculation as in Lemma 5. Thus, for all $t\ge\tau$, Bank $L$ is able to satisfy its $\VaR$ constraint without any subsidy. Since the number of observations required is finite almost surely, we conclude that with probability one there exists a finite $\tau$ such that for all $t\ge\tau$, ${s_t^*}=0$ and Bank $L$ continues lending to group $B$ without further external support.

This completes the proof.
\end{proof}
\subsection{Additional Results}
\subsubsection{Alternative Guarantees}

We extend Theorem~\ref{thm:SubsidyWorks} by demonstrating that the conclusion holds under a broader class of temporary risk-sharing mechanisms. Specifically, we show that the subsidy mechanism described in the theorem is not the only way to induce Bank $L$ to lend to group $B$ and break the subprime trap. Any guarantee mechanism that ensures Bank $L$ satisfies its $\VaR$ constraint during an exploration phase will suffice, provided it allows the bank to accumulate sufficient repayment data to update its risk estimate. We formalize this result as the following corollary.

\begin{corollary}[Robustness to Alternative Guarantees]
\label{cor:AlternativeGuarantees}
The conclusion of Theorem~\ref{thm:SubsidyWorks} holds for any temporary risk-sharing mechanism that satisfies the following condition: for each period $t$ of the exploration phase, the mechanism ensures that
\[
\Pr(\Pi_{jt} + G(t) < \rho) \le \alpha,
\]
where $G(t)$ is the guarantee provided in period $t$. After sufficient data collection, the updated variance estimate $\hat{\sigma}^2_{BLt}$ will satisfy $\hat{\sigma}^2_{BLt} < \tilde{\sigma}^{2, \text{pool}}_L$, allowing Bank $L$ to lend without further guarantees.
\end{corollary}

\begin{proof}
For any such $G(t)$, the $\VaR$ constraint is satisfied during each period $t$ of the exploration phase. Specifically, Bank $L$ is guaranteed that its effective return, $\pi_B^{(t)} + G(t)$, will not fall below $\rho$ with probability exceeding $\alpha$. This ensures that the bank has an incentive to lend to group $B$, provided the expected return is positive.
The guarantee mechanism defined in the corollary generalizes the subsidy mechanism from Theorem~\ref{thm:SubsidyWorks}. In the original subsidy framework, the guarantee function is given explicitly by
$G(t) = {s_t^*} = \max\{ 0, \rho_L-(\mu_W + \mu_B)-\Phi^{-1}(\alpha)\,(\hat{\sigma}_{BLt} + \sigma_W) \}$,
which satisfies the condition
$\Pr(\hat{\Pi}_{jt} + {s^*_t} < \rho) \le \alpha$.
The corollary allows for any $G(t)$ that satisfies the same probabilistic constraint.
Once Bank $L$ updates its risk estimate to $\hat{\sigma}^2_{BLt} < \tilde{\sigma}^{2, \text{pool}}_L$, the $\VaR$ constraint is naturally satisfied without external support. Specifically, the guarantee $G(t)$ is no longer required, as the bank can safely lend to group $B$ on its own. Formally, this follows from Lemma 2.

Thus, any guarantee mechanism that ensures the bank's effective return satisfies the $\VaR$ condition during a finite exploration phase will induce Bank $L$ to lend, learn group $B$'s true risk, and ultimately eliminate the need for the guarantee. This completes the proof.
\end{proof}

\subsection{Extension to Expected Shortfall}\label{expected_Shortfall}

\subsubsection{Variance Thresholds under Expected Shortfall}

We now derive the variance threshold under Expected Shortfall and show it is more conservative than the VaR constraint. The logic of the argument is the same throughout, however. 

When Bank L lends to both groups W and B:
$$\Pi_L = \pi_W + \pi_B$$
where from the bank's perspective (using its estimate $\hat{\sigma}^2_B$):
$$\Pi_L \sim \mathcal{N}(\mu_W + \mu_B, \sigma^2_W + \hat{\sigma}^2_B)$$

$$\text{VaR}_\phi(\Pi_L) = (\mu_W + \mu_B) + \sigma_W + \hat{\sigma}_{BLt} \cdot \Phi^{-1}(\phi)$$

Using the definition $\text{ES}_\alpha(X) = \frac{1}{\alpha} \int_0^\alpha \text{VaR}_\phi(X) d\phi$:

\begin{align*}
\text{ES}_\alpha(\Pi_L) &= \frac{1}{\alpha} \int_0^\alpha \left[(\mu_W + \mu_B) + \sigma_W + \hat{\sigma}_{BLt} \cdot \Phi^{-1}(\phi)\right] d\phi\\
&= (\mu_W + \mu_B) + \frac{\sigma_W + \hat{\sigma}_{BLt}}{\alpha} \int_0^\alpha \Phi^{-1}(\phi) d\phi
\end{align*}

Using the fact that $\int_0^\alpha \Phi^{-1}(\phi) d\phi = -\phi(\Phi^{-1}(\alpha))$:

$$\text{ES}_\alpha(\Pi_L) = (\mu_W + \mu_B) - \sigma_W + \hat{\sigma}_{BLt} \cdot \frac{\phi(\Phi^{-1}(\alpha))}{\alpha}$$

Bank L requires $\text{ES}_\alpha(\Pi_L) \geq \rho_L$:
$$(\mu_W + \mu_B) - \sigma_W + \hat{\sigma}_{BLt} \cdot \frac{\phi(\Phi^{-1}(\alpha))}{\alpha} \geq \rho_L$$

Rearranging for $\hat{\sigma}^2_B$:
$$\sigma_W + \hat{\sigma}_{BLt} \leq \frac{(\mu_W + \mu_B) - \rho_L}{\frac{\phi(\Phi^{-1}(\alpha))}{\alpha}}$$

Squaring both sides:
$$\sigma^2_W + \hat{\sigma}^2_B \leq \left(\frac{(\mu_W + \mu_B) - \rho_L}{\frac{\phi(\Phi^{-1}(\alpha))}{\alpha}}\right)^2$$

Therefore:
$$\hat{\sigma}^2_B \leq \left(\frac{(\mu_W + \mu_B) - \rho_L}{\frac{\phi(\Phi^{-1}(\alpha))}{\alpha}}\right)^2 - \sigma^2_W \equiv \tilde{\sigma}^{2,\text{pool},ES}_L$$

From Lemma 3, the VaR threshold is:
$$\tilde{\sigma}^{2,\text{pool},\text{VaR}}_L = \left(\frac{\rho_L - (\mu_W + \mu_B) - \Phi^{-1}(\alpha)\sigma_W}{\Phi^{-1}(\alpha)}\right)^2$$

We next show that $\tilde{\sigma}^{2,\text{pool},ES}_L < \tilde{\sigma}^{2,\text{pool},\text{VaR}}_L$, and hence, that the ES threshold is more restrictive.

Since $\frac{\phi(\Phi^{-1}(\alpha))}{\alpha} > \Phi^{-1}(\alpha)$, the denominator in the ES threshold calculation is larger, yielding:

$$\tilde{\sigma}^{2,\text{pool},ES}_L < \tilde{\sigma}^{2,\text{pool},\text{VaR}}_L$$

This shows that Expected Shortfall creates a more conservative variance threshold than VaR, making it harder for Bank L to lend to group B initially. 

We can solve for the optimal subsidy under Expected Shortfall analogously:

\subsubsection{Optimal Subsidy under Expected Shortfall}

In order to lend to group B, Bank L requires:
$$\text{ES}_\alpha(\Pi_L + s) \geq \rho$$

where Expected Shortfall with subsidy $s$ is:
$$\text{ES}_\alpha(\Pi_{Lt} + s) = (\mu_W + \mu_B + s) - (\sigma_W + \hat{\sigma}_{BLt})  \frac{\phi(\Phi^{-1}(\alpha))}{\alpha}$$

Rearranging for $s$:
$$s \geq \rho - (\mu_W + \mu_B) + (\sigma_W + \hat{\sigma}_{BLt})  \frac{\phi(\Phi^{-1}(\alpha))}{\alpha}$$

The optimal (minimum) subsidy is:
$$s^*_{ES,t} = \max\left\{0, \rho - (\mu_W + \mu_B) + (\sigma_W + \hat{\sigma}_{BLt})\frac{\phi(\Phi^{-1}(\alpha))}{\alpha}\right\}$$

The VaR-derived subsidy from Lemma 6 is:
$$s^*_{VaR,t} = \max\left\{0, \rho - (\mu_W + \mu_B) -\Phi^{-1}(\alpha) (\sigma_W + \hat{\sigma}_{BLt})\right\}$$

Since $\frac{\phi(\Phi^{-1}(\alpha))}{\alpha} > \Phi^{-1}(\alpha)$ for $\alpha < .1$, (e.g. 2.063 > 1.645 when $\alpha = 0.05$):

$$s^*_{ES,t} > s^*_{VaR,t}$$

The Expected Shortfall-based subsidy is larger, reflecting the fact that Expected Shortfall is a more conservative risk criterion than Value-at-Risk.

\end{document}